\documentclass[a4paper,11pt]{article}

\usepackage{amsthm,amsmath,amscd,amssymb,amsfonts,pstricks,pst-node,graphics}

\def\bbbz{{\mathbb Z}}

\def\cp1{{\mathbb C\mathbb P}^1}

\newtheorem{Def}{Definition}

\newtheorem{Thm}{Theorem}

\usepackage{amsthm}
\usepackage{amsmath}
\usepackage{amssymb}
\usepackage{amsfonts}
 \usepackage{amscd}
\def\cE{{\cal S}}

\def\bbbz{{\mathbb Z}}

\def\cp1{{\mathbb C\mathbb P}^1}

\def\cE{{\cal S}}

\newcommand\cS{{\mathcal S}}

\numberwithin{equation}{section}
\numberwithin{figure}{section}\usepackage[small]{caption}

\begin{document}
\bibliographystyle{unsrt} 
\title {Master symmetry and time dependent symmetries of the
differential-difference KP
equation
}
\author{Farbod Khanizadeh\\
School of Mathematics, Statistics \& Actuarial Science\\
University of Kent,
Canterbury, UK}
\date{}
\maketitle
\begin{abstract}
We first obtain the master symmetry of the differential-difference KP equation.Then we show how this master symmetry, through $sl(2,\mathbb{C})$-representation of the equation, can construct generators of time dependent symmetries. 
\end{abstract}
\section{Introduction}
The differential-difference KP (DDKP) equation is a lattice field equation of the following form \cite{nijhoff1984backlund, etu, nijhoff1986direct, tamizhmani1998wronskian}:
   \begin{eqnarray}\label{BO}
   u_{t}=u_{xx}+2uu_{x}+2(\cS-1)^{-1}u_{xx}=K\,,
   \end{eqnarray}
 where the dependent variable $u:=u(n,x,t)$ is a smooth function of continuous varibles $x$, $t$ and 
discrete variable $n\in \bbbz$.  We use the notation $u_{m,j}=:D_{x}^{j}u(n+m,t,x)=\cS^{m}D_{x}^{j}u(n,t,x)$ where $\cS$ is the linear shift operator and $D_{x}$ denotes the derivative with respect to $x$.
 The DDKP equation is also appeared in \cite{BS01, vel1997lax}. 
  In \cite{BS01} it is constructed through the procedure of central extension and the author refers to the equation as the lattice field Benjamin-Ono. For the links between hierarchy of equation (\ref{BO}) and the continuous Kadomtsev–-Petviashvili hierarchy one can read \cite{fu2013integrability}.

 Notice that equation (\ref{BO}) involves $(\cS-1)^{-1}$ and is no longer local.
The nonlocality is the main obstacle to directly extend the theories of (1+1)-dimensional nonlinear evolutionary equations 
to (2+1)-dimensional equations. 

Master symmetries are tools by which, through a recursive process, one can construct a hierarchy of infinitely many commuting generalised symmetries and guarantee the integrability.  The notion of master symmetry was first introduced in \cite{  mr83f:58039} and later studied in \cite{ mr86c:58158,  mr83f:58043}. The efficiency of these tools can be clearly seen in the (2+1)-dimensional case for which recursion operators of equations  have more complicated structure than in (1+1)-dimension \cite{FS88b, FS88c}.


Master symmetries of (1+1)-dimensional differential and differential-difference equations are not always local. For example the KdV
\begin{eqnarray*}
u_{t}=u_{xxx}+6uu_{x}\,,
\end{eqnarray*}
admits the following nonlocal master symmetry\cite{ dorf87, fok87, oev84}:
\begin{eqnarray}\label{mesal1}
\tau=x(6uu_{x}+u_{xxx})+4u_{xx}+8u^2+2u_{x}D^{-1}_{x}u\,.
\end{eqnarray}
Master symmetries are also related closely to  time dependent symmetries of evolutinary equations. In \cite{ mr86c:58158, mr1874291} this relation is discussed along with examples of continous case. It is shown having a master symmetry of evolution equation, one can produce an infinite number of time dependent symmetries which are polynomial in time.

In this paper our focus is on the DDKP equation. We write the equation in the form in which a quasi-local term appears. Then we introduce briefly the concept of master symmetries and derive the master symmetry of the equation. The paper ends with the application of the master symmetry in constructing generators of the time dependent symmetries of the DDKP equation.   This relation is illustrated through the sl$(2,\mathbb{C})$-representation of the  DDKP equation.

\section{The differential-difference KP equation and its master symmetry}\label{se}
We can rewrite the differential-difference KP (DDKP) equation (\ref{BO}) in the following form: 
   \begin{eqnarray}\label{BOO}
   u_{t}=u_{xx}+2uu_{x}+2\Theta(u_{x})=K\,,
   \end{eqnarray}
where $\Theta:=(\cS-1)^{-1}D_{x}$ and $\Theta(u_{x})$ is called a quasi-local function.

The concept of quasi-local functions was first introduced for the continous case in 1998 by Mikhailov and Yamilov \cite{mr1643816}.
They noticed that the operators $D_x^{-1}$ and $D_y^{-1}$ never appear alone
but always in pairs like $D_x^{-1}D_y$ and $D_y^{-1}D_x$ for
all known integrable equations and their hierarchies of symmetries
and conservation laws. 
In \cite{wang21}, using the symbolic representation, the author proved this 
observation is true for certain type of integrable equations.
Here we denote ${\hat{\mathfrak{G}}}$ for the space of quasi-local polynomials which depend on $u_{i,j}$ and contain the operator $\Theta$. We also define ${{\mathfrak{G}}}$ as the extension of ${\hat{\mathfrak{G}}}$ to the space with coefficients from the set $\{n,x\}$. For more details and discussion on the structure of  ${\hat{\mathfrak{G}}}$ and ${{\mathfrak{G}}}$ one can read \cite{Farbod1}.

For two elements $F$ and $G$ in ${{\mathfrak{G}}}$ we define the Lie bracket
 \begin{eqnarray}\label{2symm}
 [F,G]:=F_{*}(G)-G_{*}(F)\,,
 \end{eqnarray}
where $F_{*}$ is called the Fr\'{e}chet derivative of $F$ and is defined by:
 \begin{eqnarray}\label{ahan}
 F_{*}=\sum_{i,j}\frac{\partial F}{\partial u_{i,j}}\cE^{i}D_{x}^{j}\,.
 \end{eqnarray}
 The anti-symmetric axiom is clear and bilinearity is deduced from the linearity of derivations and the shift operator. One can also prove the Jacobi identity through a direct calculation.
\begin{Def}
We call an element $G\in{\hat{\mathfrak{G}}}$ a generalised symmetry of equation (\ref{BOO}) if and only if $ [K,G]=0\,.$
 Furthermore equation (\ref{BOO})  is called integrable if and only if it possesses
infinitely many commuting generalised symmetries.
\end{Def}

As mentioned in the introduction, the master symmetry is a quantity by which one can construct a hierarchy of commuting symmetries. 
 For a comprehensive description of  master symmetries one can read \cite{ mr83f:58039, mr86c:58158, fuchssteiner1985mastersymmetries,mr2000e:37104, flack1, oef82}.

\begin{Def}\label{newmaster1}
An element $W\in\mathfrak{G}$ is called a master symmetry of equation (\ref{BOO}) if and only if
\begin{eqnarray}\label{newmaster2}
[K,[K,W]]=0,\qquad [W,K]\in\hat{\mathfrak{G}}\,.
\end{eqnarray}
\end{Def}
 In practice we are interested to find a $W$ such that  its recursive action on $K$ is well-defined and produces a new element in $\hat{\mathfrak{G}}$. 
In the following theorem we provide the explicit form of the master symmetry of equation (\ref{BOO}). 
\begin{Thm}\label{miss}
The element $W\in\mathfrak{G}$ of the form
\begin{eqnarray}\label{BOmaster}
W=xu_{xx}+2xuu_{x}+2x\Theta(u_{x})+nu_{x}+u^2+3\Theta(u)\,,
\end{eqnarray}
is the master symmetry of the DDKP equation (\ref{BOO}). Furthermore the recursive action of $W$ on the equation produces a commutative space of symmetries.
\end{Thm}
\begin{proof}
According to Definition \ref{newmaster1}, to say $W$ is a master symmetry we need to show $[K,[W,K]]=0$ where $[W,K]$ lies in the space $\hat{\mathfrak{G}}$. This implies that $[W,K]$ should be a generalised symmetry of equation (\ref{BOO}). The first non-trivial symmetry of (\ref{BOO}) is given in \cite{BS01}. We rewrite it in our notation as follows:
\begin{eqnarray}\label{BOsymmetry}
G&=&u_{xxx}+3uu_{xx}+3u_{x}^2+3u^2u_{x}+3\Theta(uu_{x})+3u_{x}\Theta(u)\nonumber\\
&+&3u\Theta(u_{x})+3\Theta^2(u_{x})+3\Theta(u_{xx})\,.
\end{eqnarray}
From (\ref{ahan}) one can check that the bracket $[K,G]$ vanishes and $[W,K]=-2G$, which proves the first part.

The second part of the theorem is thoroughly discussed in \cite{Farbod1}. There, the author refers to the theorem namely $\tau$-scheme by I. Dorfman \cite{mr94j:58081} and shows how $W$ meets the conditions of  the $\tau$-scheme. 
\end{proof}
In the next section we give the description of an irreducible $sl(2,C)$-representation by which we can construct the time dependent symmetries of the DDKP equation.

\section{Time dependent symmetries}\label{panj}
 In \cite{mr1874291} it was shown how the representation of Burgers, Ibragimov-Shabat and KP equation can help one to construct the time dependent symmetries.
  In what follows, using the Lie algebra $sl(2,\mathbb{C})$  we present such a construction for equation (\ref{BOO}). We call $G(u,t)$ a time dependent symmetry of (\ref{BOO}) if and only if
\begin{eqnarray}\label{tdep}
\frac{\partial G}{\partial t}=[G,K]\,.
\end{eqnarray}
We find an appropriate element $G_{0}\in\hat{\mathfrak{G}}$ such that for some $\ell$ we get $ad_{K}^{\ell}(G_{0})=0$, where
\begin{eqnarray*}
 ad_{K}^{\ell}G_{0}:=[\underbrace{K\cdots,[K,[K}_{\ell-times},G_{0}]]]\,,
\end{eqnarray*}
 Then it can be shown $G(u,t)=\exp(-t ad_{K})(G_{0})$ is a time dependent symmetry of equation (\ref{BOO}) \cite{mr86c:58158, mr1874291}. To begin with, we consider
\begin{equation}\label{sl2}
M=K,~~~N=-\frac{1}{2}x,~~~H=-xu_{x}-u\,.
\end{equation}
Calculating the Lie brackets yields
\begin{equation}\label{relation}
[M,N]=H,~~~[H,N]=-2N,~~~[H,M]=2M\,.
\end{equation}
Therefore we have the following Lie algebra isomorphisms:
$$sl(2,C)\cong span\{M,N,H\}$$
We know from the representation theory of $sl(2,\mathbb{C})$ (more details can be found in \cite{ springerfarbod, humphreys1972introduction}) that if $V$ is a finite dimensional irreducible representation, then there is a vector $\upsilon_{0}\in V$ where $H\upsilon_{0}=\lambda v_{0}$ and $M\upsilon_{0}=0$ and furthermore $V$ is spanned by the $\lambda+1$ linear independent vectors $\{\upsilon_{0},ad_{N}\upsilon_{0},ad_{N}^2\upsilon_{0},\cdots,ad_{N}^{\lambda}\upsilon_{0}\}$. 
 The element $\upsilon_{0}$ is called the highest weight vector.
 
Since $[M,K]=0$ and $[H,K]=2K$, by setting $\upsilon_{0}=K$ we obtain a three dimensional representation of $sl(2,\mathbb{C})$ spanned by $V_{3}=span\{K,ad_{N}K,ad^{2}_{N}K\}$. 
  Now let $N_{2}$ be the master symmetry of the DDKP equation. Hence
 $$[M,ad^{m}_{N_{2}}K]=0\,.$$
 The following theorem provides the basis of  $sl(2,\mathbb{C})$-representations with different dimension. For the proof one can read \cite{mr1874291}.
 \begin{Thm}
 For elements defined in (\ref{sl2}) the following relation holds:
 \begin{eqnarray*}
 [H,ad^{m}_{N_{2}}K]=(m+2)ad^{m}_{N_{2}}K\,.
 \end{eqnarray*}
 \end{Thm}
This theorem implies that for different value of $m$, $ad^{m}_{N_{2}}K$ is a highest weight vector of the $(m+3)$-dimensional representation $$V_{m+3}=span\{ad^{m}_{N_{2}}K,ad_{N}ad^{m}_{N_{2}}K,ad^{2}_{N}ad^{m}_{N_{2}}K,\cdots,ad_{N}^{m+2}ad^{m}_{N_{2}}K\}\,,$$
 where
 \begin{eqnarray*}
 ad^{m_2}_{N}ad^{m_1}_{N_{2}}K=[\underbrace{N,[N,\cdots,[N}_{m_{2}-times},[\overbrace{N_2,[N_2,\cdots[N_{2}}^{m_{1}-times},K]]]]]]\,.
 \end{eqnarray*}
 Therefore starting from $u_{x}$ we can draw the following diagram for the DDKP equation in which the first element of each row is the symmetry of equation (\ref{BOO}) and the horizontal lines contain the representations of $sl(2,\mathbb{C})$:\vspace{10mm}\\

$\begin{CD}
u_{x}@>N>>ad_{N}u_{x}\\
@VVN_{2}V\\
K @>N>>ad_{N}K @>N>>ad^{2}_{N}K\\
@VVN_{2}V \\
ad_{N_{2}}K @>N>> ad_{N}ad_{N_{2}}K @>N>> ad^{2}_{N}ad_{N_{2}}K @>N>> ad^{3}_{N}ad_{N_{2}}K \\
 @VVN_{2}V\\
ad^{2}_{N_{2}}K @>N>> ad_{N}ad^{2}_{N_{2}}K  @>N>> ad^{2}_{N}ad^{2}_{N_{2}}K  @>N>> ad^{3}_{N}ad^{2}_{N_{2}}K  @>N>> ad^{4}_{N}ad^{2}_{N_{2}}K \\
@VVN_{2}V\\
 \vdots
\vspace{5mm}
\end{CD}$
\vspace{5mm}\\
In the following theorem we see that horizontal lines are not just elements of the representation but  generators for time dependent symmetries of equation (\ref{BOO}).
  \begin{Thm}
   Consider $ad_{N_{2}}^mK\in V_{m+3}$, then the basis elements
   $$ad_{N}^{\ell}ad_{N_{2}}^{m}K\,,$$
 commute with $K^{\ell+1}, (0\leq \ell\leq m+2)$. In other words we have 
  \begin{eqnarray*}
  ad^{\ell+1}_{K}(ad_{N}^{\ell}ad^{m}_{N_{2}}K)=0\,.
  \end{eqnarray*}
  \end{Thm}
  \begin{proof}
  The proof is by induction on $\ell$. For $\ell=0$ as all $ad_{N_{2}}^{m}K$ are symmetries of (\ref{BOO}), we have
 $$[K,ad_{N_{2}}^mK]=0\,.$$
Assume the statement holds for
$$ad_{K}^{\ell+1}(ad_{N}^{\ell}ad_{N_{2}}^{m}K)=0\,,$$
then we have
\begin{eqnarray*}
ad_{K}^{\ell+2}ad_{N}^{\ell+1}ad_{N}^{m}K&=&[K,ad_{K}^{\ell+1}[N,ad_{N}^{\ell}ad_{N_{2}}^{m}K]]  \\&=&[K,[ad_{K}^{\ell+1}N,ad_{K}^{\ell+1}ad_{N}^{\ell}ad_{N_{2}}^{m}K]]=0\,.
\end{eqnarray*}
  \end{proof}


 Hence each basis element, lying on the horizontal line, provides a generator for the time dependent symmetries of the  DDKP equation. Bellow we give two examples of such symmetries 
\begin{eqnarray*}
&&\exp(-tad_{K})(ad_{N}K)=[K,N]-t[K,[K,N]]=(-xu_{x}-u)+2t(u_{xx}+2uu_{x}+2\Theta(u_{x}))\,,
\\&&\exp(-tad_{K})(ad^2_{N}K)=ad^{2}_{N}K-t(ad_{K}ad_{N}^2K)+\frac{t^2}{2!}(ad_{K}^2ad_{N}^{2}K)=x-2t(xu_{x}-u)\\&&+4t^2(u_{xx}+2uu_{x}+2\Theta(u_{x}))\,.
\end{eqnarray*}

\section{Conclusion}\label{shish}
In this paper we showed how the master symmetry of the DDKP equation can produce an infinite hierarchy of commuting symmetries. We also demonstrated how systematically, using the Lie algebra $sl(2,\mathbb{C})$, one can obtain time dependent symmetries that are polynomial in $t$. The Lie algebra $sl(2,\mathbb{C})$ is constructed around the scaling symmetry and the master symmetry of a given equation. This structure can be also used for (1+1)-dimensional equations in which the equation, symmetries and the basis elements of $sl(2,\mathbb{C})$ are all local. 

\section*{Acknowledgements}
The author would like to express his sincere appreciation to his supervisor Dr J. P. Wang for her guidance and advice throughout this paper. I would like to thank Dr. G. Papamikos for many helpful discussions and comments. This work was also carried out during my PhD study in SMSAS at the University of Kent and I am grateful for supports received from the School.

\bibliography{kdv}
\end{document}